\documentclass[11pt]{article}

\DeclareUnicodeCharacter{2212}{-}
\usepackage{amsmath}
\usepackage{amssymb}
\usepackage{amsthm}
\usepackage{amsmath}

\usepackage{float}
\usepackage[numbers]{natbib}
\usepackage{graphicx}
\usepackage{a4wide,xcolor}
\usepackage[unicode]{hyperref}
\hypersetup{
    colorlinks=true,
    linkcolor=blue,
    filecolor=magenta,      
    urlcolor=cyan,
}

\newtheorem{theorem}{Theorem} 
\newtheorem{defi}{Definition}
\newtheorem{corollary}{Corollary}
\newtheorem{proposition}{Proposition}
\newtheorem{lemma}{Lemma}

\usepackage{tikz}
\newcommand{\Cross}{\mathbin{\tikz [x=1.4ex,y=1.4ex,line width=.2ex] \draw (0,0) -- (1,1) (0,1) -- (1,0);}}%

\title{Extremal combinatorics,
iterated pigeonhole arguments\\
and generalizations of PPP}
\author{\textbf{Amol Pasarkar} \\ Columbia University \\ amol.pasarkar@columbia.edu
\and \textbf{Christos Papadimitriou} \\ Columbia University \\ christos@columbia.edu 
\and \textbf{Mihalis Yannakakis} \\ Columbia University \\ mihalis@cs.columbia.edu

}
\date{}

\begin{document}

\maketitle

\begin{abstract}
We study the complexity of computational problems arising from existence theorems in extremal combinatorics. For some of these problems, a solution is guaranteed to exist based on an \textit{iterated} application of the Pigeonhole Principle.  This results in the definition of a new complexity class within TFNP, which we call PLC (for ``polynomial long choice"). PLC includes all of PPP, as well as numerous previously unclassified total problems, including search problems related to Ramsey's theorem, the Sunflower theorem, the Erd\H os-Ko-Rado lemma, and K\" onig's lemma.  Whether the first two of these four problems are PLC-complete is an important open question which we pursue; in contrast, we show that the latter two are PPP-complete. Finally, we reframe PPP as an optimization problem, and define a hierarchy of such problems related to Tur\` an's theorem. 
\end{abstract}
	
\section{Introduction}
The complexity class TFNP \citep{MEGIDDO_TFNP_1991} captures a wide variety of search problems which are believed to lie between P and NP --- in fact, almost all problems in NP not yet known to be in P, or close to it, appear to belong to this class. The `TF' in TFNP indicates that it is a class of total function problems --- computational search problems which are mathematically guaranteed to have a solution on all instances --- while the letters ``NP" in TFNP signify that solutions are polynomially checkable.  Two things make TFNP interesting:  First, it is a microcosm of many complexity classes, each of which is identified with a non-constructive combinatorial lemma used in the proof of totality \citep{PAPADIMITRIOU1994},\citep{JOHNSON1988}, \citep{GOLDBERG2018}.  Second, it contains --- almost by its definition ---  many problems that are of great interest in Cryptography.  The most obvious and best known example is of course factoring, but many other computational problems of cryptographic interest lie in the subclass \textbf{PPP} \cite{JERABEK2016380,Sotiraki2018} whose existence lemma is the Pigeonhole Principle: \textit{if there are $2^n$ pigeons to be placed into $2^n - 1$ pigeonholes, there must exist a pigeonhole with at least two pigeons.}


\textit{In this paper we study the complexity of total search problems in the important field of extremal combinatorics} \cite{bollobas,graham1990ramsey,jukna}.  Notice that the Pigeonhole Principle itself can be seen as an argument in extremal combinatorics: \textit{``If a combinatorial object of a certain kind (here, a set mapped to $[N]$) is large enough, it must contain a certain substructure (here a collision).''} This leads one to ask: Are then the computational problems coming from extremal combinatorics in PPP?  
We point out that the combinatorial lemma underlying many such problems is a counting argument which iteratively uses a form of the Pigeonhole Principle. Accordingly, we introduce a new complexity class, which we call ``Polynomial Long Choice" (PLC) capturing the complexity of the iterated Pigeonhole Principle. 

To understand the generic problem in this class, consider the following two-player game. Player 1 seeks to construct a long sequence of pigeons and Player 2 tries to make this impossible. We start with $2^n$ pigeons. At each stage, Player 1 can pick (without replacement) a single pigeon from the remaining available pigeons to add to the long sequence. Once Player 1 has made a move, Player 2 can then partition the remaining pigeons into two groups. Next, Player 1 will pick a pigeon from one of these groups to add to the sequence, and the pigeons from the other group will immediately be removed from the game. Player 1 wins if a sequence of pigeons of length $n+1$ is constructed, otherwise Player 2 wins.  It is easy to see that Player 1 has a winning strategy in this game (pick any pigeon from the larger group in every iteration) --- and this is the existence lemma defining PLC.  To make this into a computational problem, which we call {\sc Long Choice}, we equip Player 2 with a suite of polynomial-time algorithms, one for each stage of the game (we give the precise definitions below); PLC is the class of all search problems reduced to {\sc Long Choice}.

We show that {\sc Long Choice} is PPP-hard, and hence PLC contains PPP.  But is this containment strict?  And even if not, why is {\sc Long Choice}, defined by an iterative application of pigeonhole arguments, not contained in P$^{\rm PPP}$?  The difficulty is this:  The winning strategy requires that Player 1 estimates the majority correctly at each stage and also successfully selects a pigeon from this majority. PPP does not seem to support both of these challenges.

It turns out that PLC contains a host of natural problems embodying important theorems in extremal combinatorics, first and foremost Ramsey's, but also the sunflower theorem, the Erd\H os-Ko-Rado lemma and K\"onig's Lemma. The latter two, however, can be shown to be PPP-complete.  We also study  problems associated with another classical result in extremal combinatorics, namely Mantel's Theorem (``a graph with $N$ nodes and more than $N^2/ 4$ edges cannot be triangle-free'') and Turan's theorem (the generalization to $k$-clique-free graphs). We identify an infinite hierarchy of problems related to these theorems, each of which is PPP-hard. This generalization of PPP seems substantially different from PLC, in the sense that the source of computational hardness arises from information-theoretic reasons, namely, an inefficient encoding of the search problem. 

But iterating the Pigeonhole Principle can take us even higher:  consider the dual problem which can be called {\sc Short Choice}:  Suppose that the above game had $2^n - 2$ pigeons, and that Player 1 now
wants to terminate the game as soon as possible and Player 2 wants the opposite, to extend it;
the game terminates when one of the groups created by Player 2 is empty. 
It is easy to see that Player 1 cannot be forced to make more than $n-1$ moves, by choosing in every iteration to continue in the smaller of the two groups.  Now call this problem {\sc Short Choice;} it is certainly total, but it does not seem to belong to NP (how does one verify that there is no pigeon left that is consistent with all the previous choices of Player 1?). This is reminiscent of the {\em empty pigeonhole principle} recently explored in \cite{Kleinberg2020TFNP} and the class PEPP belonging in $TF\Sigma_2P$ and not believed to be inside TFNP (or NP).  
We show that {\sc Short Choice} is a PEPP-hard problem that defines a new subclass of $TF\Sigma_2P$. 

\section{Long Choice}
We start by recalling the definition of the class PPP.  We first define the problem {\sc Collision} to be the following: we are given a Boolean circuit $C$ with $n$ input bits and $n$ output bits, and we seek either (a) an input $x$ such that $C(x)=0^n$, or (b) a {\em collision}, two distinct inputs $x\neq y$ such that $C(x)=C(y)$.  The class PPP is the set of all search problems that reduce to {\sc Collision.}
For the {\em weak} version of PPP, denoted PWPP, the circuit has $n$ inputs and $n-1$ outputs, and a collision is sought.  It is known that $n-2$ or fewer outputs, down to $n^\delta$ for any $\delta>0$, yield the same class \cite{Komargodski2019}.  

Let us next define the search problem {\sc Ramsey,} motivated by one of the most influential theorems in all of combinatorics: Given a graph with $2^{2n}$ nodes, represented by a circuit with $4n$ input bits
and one output bit, we seek either a clique with $n$ nodes, or an independent set with $n$ nodes. 
The nodes are represented by $2n$-bit strings and the circuit specifies the edge relation of the graph.
The well known proof of Ramsey's theorem proceeds by constructing a sequence of $2n$ nodes, where: the first node is arbitrary; and the next node is selected from the available nodes to belong in the {\em majority,} either adjacent or nonadjacent to the last node, whichever group is larger.  In addition, the smaller group becomes unavailable.  Since we start with $2^{2n}$ nodes and the minority becomes unavailable at each step, it is clear that a sequence of $2n$ nodes can be selected, and therein we will find either an independent set or a clique with $n$ nodes.

This proof inspires the key definition of this paper:
\begin{defi} 
The {\sc Long Choice} problem is the following: There is a universe $U$ of $2^n$ objects, represented by the $2^n$ $n$-bit strings.  We are given a sequence of $n-1$ circuits $P_0, \dots, P_{n-2}$, each of poly($n$) size, such that $P_i$ has $(i+2)n$ input bits and one output bit;  circuit $P_i$ represents a {\em predicate} on $i+2$ objects.  We are asked to find a sequence of $n+1$ distinct objects $a_0, \dots, a_{n}$, with the following property: for each $i$ in $[0, \dots, n-2]$, $P_i(a_0, \dots, a_{i},a_{j})$ is the same for all $j>i$.
\end{defi}

\begin{theorem}
{\sc Long Choice} is a total problem in TFNP.
\end{theorem}

\begin{proof}
    Our construction is inspired by the proof of Ramsey's Theorem as well as the two-player game which we described in the introduction. 
    
    First, we pick an arbitrary element $a_0$ in the universe. Then, we partition the remaining elements $a_i$ into two categories, based on the value of $P_0(a_0, a_i)$. Since there are $2^n - 1$ elements being partitioned, the majority of this partition must have at least $2^{n-1}$ elements. We select an arbitrary element $a_1$ from this majority and discard all elements from the minority. 
    
    We then continue this procedure: we partition the remaining elements $x$ based on the value of $P_1(a_0, a_1, x)$, and we pick an arbitrary element $a_2$ from the majority of this new partition. 
    
    We can continue partitioning elements in this fashion and picking elements from the majority until we arrive at a complete Long Choice certificate. This proves the totality of the problem. Membership in TFNP follows from the fact that a candidate certificate sequence $a_0, \ldots a_n$ can be checked easily in polynomial time.
\end{proof}

The proof implies that the problem remains total if the sequence $a_0, a_1, \ldots, a_n$ is constrained to start with a specific given object $a_0$. In fact, the constrained variant of the problem where the starting object $a_0$ is specified in the input of the {\sc Long Choice} problem turns out to be polynomially equivalent to the unconstrained version defined above; a proof is given in the appendix.

Critically, {\sc Long Choice} is also PPP-hard. We give first an outline of the basic idea of the proof, and then proceed to the detailed formal proof.  Consider the two-player game from the Introduction which characterizes {\sc Long Choice}. In this game, player 2 can behave (i.e., the predicates $P_i$ can be specified) in a way that guarantees that the only way player 1 wins is by finding a certificate to a PPP-complete problem {\sc Collision.} Consider an instance of {\sc Collision}, given by a circuit, $C$, which maps $n$-bit strings to $n$-bit strings. Player 1 starts the game with $2^n$ objects (the domain of the circuit). At each round, player $1$ picks an element from the remaining set of objects. If at any round, player $1$'s choices so far contain a certificate to {\sc Collision} (that is, one element is a zero element or a pair of elements collide under $C$), then player 2 stops partitioning the remaining elements. That is, player 2, for the rest of the game, places all remaining elements into the same side of the partition, guaranteeing a path to victory for player 1. 

In each round, player $2$ considers the ``vacant spots" in the range of $C$: the nonzero values of the range that do NOT contain the image of player $1$'s choices. 
At round $1$ of the game, player $1$ makes some choice, call it $a_0$. If $C(a_0) = 0$, then player $1$ has found a certificate, and we are done. If $C(a_1)$ is positive, then there are $2^n - 2$ vacant spots left. Player $2$ splits the set of vacant spots into two even halves (there are many ways to do this, one way is for player $2$ to specify a constant, $k$ and declare that all vacant spots less than or equal to $k$ belong to one half, and the vacant spots greater than $k$ belong to the other). 

Regardless of what value player $1$ chooses for $a_1$ in round $2$, all subsequent elements must belong to the same subgroup of $C(a_1)$. After $C(a_1)$ is chosen, the number of vacant spots in this subgroup is $2^{n-1} - 2$. 

The game continues in this fashion, with player $2$ always taking note of the remaining available vacant spots, and splitting this set into $2$ even groups. In general, after the $i$-th round, there will be at most $2^{n-i+1} - 2$ vacant spots left. Therefore, after the $n$-th round, assuming no certificate has yet been found, there will be $2^{1} - 2 = 0$ vacant spots. Therefore, the $(n+1)$-th choice player $1$ makes must provide a collision (or zero element).

\begin{theorem}
{\sc Long Choice} is PPP-hard.
\end{theorem}

\begin{proof}
Suppose that we are given a circuit $C_0$ mapping $n$ bits to $n$ bits,  an instance of the {\sc Collision} problem. We can view the inputs and outputs of circuits both as $n$-bit strings or as the equivalent integers in $[0, \ldots, 2^n -1]$.  Given $C_0$, we define a new circuit, $C$ which maps $n$-bit strings to $n$-bit strings. On input $a$, we define $C$ as follows:  
	
	\begin{enumerate}
	    \item $C(a) = 2^n-1$ (the all-1 string) if $C_0(a) = 0$
	    \item $C(a) = C_0(a)$ otherwise
	\end{enumerate}
	
	By the pigeonhole principle, since $C$ only maps inputs to nonzero values, it must have collisions. Any such collision will either allow us to recover a collision in $C_0$ or allow us to recover a zero element of $C$. 
	
	We now reduce the problem instance given by circuit $C$ to a Long Choice problem. We begin by defining our set, $U$, to be the domain of circuit $C$: the set of $n$-bit strings. Note that $U$ has $2^n$ distinct elements, each of which is represented as a unique $n$-bit string. It now suffices to define $n-1$ predicate functions $P_0, P_1, \dots, P_{n-2}$.  
	
	As in the proof of totality for Long Choice, we can think about constructing a certificate by sequentially making the choices $a_0, a_1, \dots, a_n$. Put simply, our predicate functions will classify the elements of $U$ based on their \textit{images} under $C$. These functions will enforce the following property: subsequent elements of our Long Choice certificate will have images under $C$ which are closer and closer together. More specifically, consider the first three elements of a Long Choice certificate, $a_0, a_1, a_2$. While $C(a_0)$ and $C(a_1)$ may be more than $2^{n-1}-1$ units apart, predicate function $P_0$ will enforce the condition that $C(a_1)$ and $C(a_2)$ are within $2^{n-1}-1$ units of each other. Each predicate function will, in effect, enforce similar ``closeness" conditions. Ultimately, this will force any Long Choice certificate to have two distinct elements whose images under $C$ are $0$ units apart (a collision under $C$!), as desired.
	
	In the following discussion, we assume all inputs to predicate functions are distinct, because this is required of any valid certificate.
	
	We now explicitly define the predicate functions:
	\begin{enumerate}
		\item $P_i(a_0, \dots, a_i, x) = 1$ if $C(x)$ is in the interval $F_{i}$ (defined below)
		\item $P_i(a_0, \dots, a_i, x) = 0$ if $C(x)$ is \textit{not} in the interval $F_{i}$.
	\end{enumerate}
	
	To complete the definition, we define the intervals, $F_i$, which depend on the elements $a_0, \ldots, a_i$. Before we do so, we introduce some basic terminology to make this discussion clearer.

	\begin{enumerate}
		\item \textbf{Unfilled set:} For any sequence of elements $a_0, \dots, a_i$ and an interval $[p, q]$, the unfilled set of the interval is defined as 
		
		$$\{ \{p,p+1, \dots, q\}\setminus{\{C(a_0), C(a_1), \dots, C(a_i)\}}\}$$
		
		For example, if given a sequence of points $a_0, a_1, a_2$ with $C(a_0) = 0$, $C(a_1) = 2$, $C(a_2) = 1$, the unfilled set for the interval  $[0, 5 ]$ is  $\{3, 4, 5\}$. 
		
        \item \textbf{Indexing an interval: }(Purely for notational convenience) Given an interval $I = [a,b]$, for $k > 0$, define $I[k]$  as the interval containing the smallest $k$ elements of $I$, and define $I[-k]$ as the interval containing the largest $k$ elements of $I$. For example, given $I = [1,4]$, $I[2] = [1,2 ]$, the first two integers in the interval, while $I[-3] = [2, 4]$, the largest $3$ integers in the interval.
	\end{enumerate}

	We now define a sequence of intervals $B_0, \dots, B_i, \dots$ and $F_0, \dots, F_i, \dots$. Critically, each interval $F_i$ is a contained in the corresponding interval $B_i$. We proceed with an inductive definition.

	\textbf{Base Case Definition.}
	For any single element sequence $a_0$, the corresponding interval $B_0$ is $[1, 2^n - 1]$. The unfilled set for $B_0$ has size $2^n - 2$. Consider the value of $k$ that guarantees that $B_0[k]$ has an unfilled set of size $2^{n-1}-1$. This can be easily computed:
	$k=2^{n-1}-1$ if $C(a_0)> 2^{n-1}-1$, else $k=2^{n-1}$. We define $F_0 = B_0[k]$.
	\\[12pt]
	\textbf{Inductive Definition.}
	Suppose that we have a sequence of elements $a_0, \dots, a_i$, and for all $k<i$, $B_k$ and $F_k$ are defined. We first define $B_i$ using the following rules: 
	\begin{enumerate}
	    \item First, if $B_{i-1}$ is an interval of size $1$, then $B_i = F_i = B_{i-1}$
	    \item Otherwise, if $C(a_i)$ is in $F_{i-1}$, then $B_i = F_{i-1}$. If $C(a_i)$ is not in $F_{i-1}$, then $B_i = B_{i-1}\setminus{F_{i-1}}$.
	    \item Finally, we define $F_i$. Let $x$ denote the size of the unfilled set of $B_i$. Let $k$ be the smallest integer such that $B_i[k]$ has $\lceil \frac{x}{2} \rceil$ unfilled spots. We let $F_i = B_i[k]$. 
	\end{enumerate}

	This completes the definition of the Long Choice problem instance. It remains to prove that a valid certificate for this problem instance allows us to recover a collision under $C$.
	Let $a_0, a_1, \ldots, a_n$ be a valid certificate. If there is a collision among these elements, we are done.
	So assume there is no collision; we will derive a contradiction.
	
	Consider the sequence of set $B_i, F_i$ generated from the sequence $a_0, a_1, \ldots, a_n$.
	Note that $B_0\supseteq{B_1}\dots\supseteq{B_i}\dots$ and $F_i \subseteq B_i$ for all $i$.
	It is easy to show inductively from the construction that the following two properties hold:\\
	1. $\forall i, \forall j \geq i, C(a_j) \in B_i$.\\
	2. $|B_i \setminus \{ C(a_0),\ldots, C(a_i)\}| = 2^{n-i}-2$ for all $i \leq n-2$.
	
	The basis case ($i=0$) for both properties is trivial. The induction step for property 1 follows from the fact that $P_{i-1}(a_0,\ldots,a_{i-1},a_j)$ has the same value for all $j \geq i$, hence either all these $C(a_j)$ are in $F_{i-1}$ or they are all not in $F_{i-1}$ and thus they are in $B_{i-1} \setminus F_{i-1}$ (because they are all in $B_{i-1}$ by the induction hypothesis). It follows from the definition of $B_i$ that they are all in $B_i$.
	
	For the induction step of property 2 note that the induction hypothesis 
	
	$$|B_{i-1} \setminus \{ C(a_0),\ldots, C(a_{i-1})\}| = 2^{n-i+1}-2$$ 
	
	implies that both $F_{i-1}$ and $B_{i-1} \setminus F_{i-1}$ have $2^{n-i}-1$ unfilled spots. Since $C(a_i ) \in B_i$ and $a_i$ does not collide with any earlier $a_j$, it follows that $|B_i \setminus \{ C(a_0),\ldots, C(a_i)\}| = 2^{n-i}-2$.
	
	From property 2, $B_{n-2}$ has an unfilled set of size $2^2 - 2 = 2$. The interval $F_{n-1}$ is, by definition, constructed such that $F_{n-1}$ and $B_{n-2}\setminus{F_{n-1}}$ both have $1$ unfilled spot. Critically, based on the definition of $P_{n-2}$, we know that $C(a_{n-1})$ and $C(a_{n})$ must both belong to $F_{n-1}$ or both belong to $B_{n-2}\setminus{F_{n-1}}$. Therefore,
	$C(a_{n-1})$ and/or $C(a_{n})$ must collide with each other or with another element in the sequence, a contradiction. 
\end{proof}

Several problems reduce to simplified cases of Long Choice where the predicates have fixed arity.
Define {\sc Unary Long Choice} to be the version of Long Choice where every predicate $P_i$ depends only on its last argument,
i.e., $P_i(a_0,\ldots, a_i,x) = P_i(x)$.
Define {\sc Binary Long Choice} to be the version of Long Choice where every predicate $P_i$ depends only on two of its arguments, the last argument $x$ and one of the previous $a_k$, i.e. $P_i(a_0,\ldots, a_i,x) = P_i(a_k,x)$ for some $k \leq i$.

It is easy to see that PWPP reduces to {\sc Unary Long Choice}:
Given a circuit $C$ for PWPP with $n$ input bits and $n-1$ output bits, define $P_i(x)$ to be the $(i+1)$-th bit of $C(x)$. Then in any valid certificate $a_0, \ldots, a_{n-1}, a_n$ for this instance of {\sc Unary Long Choice}, we must have $C(a_{n-1}) = C(a_n)$.

\begin{theorem}
	PWPP reduces to {\sc Unary Long Choice}.
\end{theorem}

\begin{proof}
	We are given an instance of PWPP($n, n-1$), defined by a circuit $C$: 
	
	$$C: \{0,1\}^n \mapsto \{0,1\}^{n-1} $$
	
	We wish to find a collision. 
	
	We construct a Unary Long Choice problem instance with  universe $\{0,1\}^n$.  For each $i$, define $P_i(a_0,\ldots,a_i,x)$ for every tuple of arguments to be the $(i+1)$-th bit of $C(x)$; thus, the value depends only on the last argument $x$.

	Consider a valid certificate $a_0, \ldots, a_n$ for this instance. Each predicate function $P_i$ enforces a condition on the elements $a_j$, where $j > i$. In particular, $P_i$ requires that for every $j, k > i$, the $(i+1)$-th bits of $C(a_j)$ and $C(a_k)$ agree. Now, consider the last two elements of our certificate: $a_{n-1}$ and $a_n$. Following the above reasoning, the predicate functions $P_0, \dots, P_{n-2}$ enforce that $C(a_{n-1})$ and $C(a_n)$ agree over all of their $n-1$ bits, implying that they collide. Thus, a certificate for our Long Choice problem provides us with a valid certificate for the original PWPP-complete problem.
\end{proof}

In the next section we will see that Ramsey problems reduce to {\sc Binary Long Choice}.

In the opposite direction, we can make the Long Choice problem harder by requiring the elements $a_i$ in the certificate to satisfy additional conditions, while still preserving the totality of the problem; for example we can require the $a_i$ to satisfy a given total order. In the proof of the totality of {\sc Long Choice}, we partition in each step the currently available set and pick an arbitrary element from the majority. We can instead pick a specific element, e.g. the smallest element under the given ordering. We call this generalization {\sc Long Choice with Order}; see the appendix for a formal definition.

\section{Long Choice and r-Color Ramsey}
Ramsey's theorem for multi-colored graphs states that for every number $r \geq 2$ of colors and every integer $n \geq 2$, there is a number $R(r,n)$ such that for every $r$-coloring of the edges of the complete graph on $R(r,n)$ nodes there is a monochromatic clique with $n$ nodes.   The standard Ramsey theorem corresponds to the case of $r=2$ colors.  A simple proof of the multi-colored Ramsey theorem uses the same type of iterative process as the $r=2$ case, except that in every step we partition the set of available nodes into $r$ groups instead of 2\footnote{Along the same lines, we could extend Long Choice to allow the functions $P_i$ to have a more general range $[r]$ instead of $\{0,1\}$.}; as before, we pick the largest group to continue the process. 
The bound on $R(r,n)$ from this simple proof is $R(r,n) \leq r^{rn}$. Obtaining better upper and lower bounds on $R(r,n)$ for $r=2$ and for general $r$ has been (and continues to be) the subject of a long line of intense research effort.

In the computational version of the problem, denoted {\sc $(r,n)$-Ramsey}, we are given the $r$-coloring of an exponentially large complete graph, which is specified via a poly-size circuit $C$, and are asked to find a monochromatic clique of size $n$.  The {\sc Ramsey} problem of the last section is equivalent to {\sc $(2,n)$-Ramsey}.
The number $r$ of colors in general need not be fixed, it could be a function of $n$.
Assume for simplicity that $r$ is power of 2 (otherwise, replace $\log r$ in the following by $\lceil \log r \rceil$).
Every node of the complete graph is represented by a unique $rn \log r$-bit string, and the given circuit $C$ takes as input two  $rn \log r$-bit strings (two nodes $u,v$) and outputs a $\log r$-bit string (the color of the edge $(u,v)$).

\begin{theorem}
	{\sc $(r,n)$-Ramsey} is in PLC for all $r, n$.
	In particular, {\sc $(r,n)$-Ramsey} reduces to {\sc Binary Long Choice}.
\end{theorem}

\noindent
{\em Proof Sketch:} We describe first the proof for the case $r=2$, which simply follows the existence proof sketched in the previous section:
Given a 2-colored complete graph on $2^{2n}$ nodes, specified by a given circuit $C$, define the predicate $P_i$ for each $i$, to map any sequence $a_0, \ldots, a_i, x$ of nodes to the color, 0 or 1, of the edge $(a_i,x)$. Note that $P_i$ depends only on the last two arguments $a_i,x$.
Consider a valid certificate $a_0, \ldots, a_{2n}$ of this instance of Binary Long Choice.
For each $i =0, \ldots, 2n-2$, all edges $(a_i,a_j)$ for $j>i$ must have the same color, 0 or 1; assign this color to node $a_i$. At least $n$ of the $2n-1$ nodes $a_0, \ldots a_{2n-2}$ are assigned the same color.
These nodes induce a monochromatic clique of size $n$.

In the case of general $r$, given an $r$-colored
complete graph on $2^{rn \log r}$ nodes, define each function $P_i$ as follows.
Let $k= (\lfloor i/ \log r \rfloor) \log r$,
i.e, $k$ is the greatest multiple of $\log r$ that is $\leq i$.
Set $P_i(a_0, \ldots, a_i,x)$ to be the $(i-k+1)$-th bit of the color of the edge $(a_k,x)$.
Note that again all these predicates depend only on two arguments, $a_k$ and $x$. 

Consider a valid certificate $a_0, \ldots, a_{rn \log r}$ of this instance of Long Choice,
and let $b_i = a_{i \log r}$ for each $i$.
From the construction of the Long Choice instance,
it is easy to see that,
for each $i=0, \ldots , r(n-1)$, all edges $(b_i,b_j)$
for $j>i$ must have the same color:
note that the $t$-th bit of the color of edge $(b_i,b_j)$, for all $j>i$, is the value of $P_{l}(a_0, \ldots, a_l, b_j) = P_{l}(a_0, \ldots, a_l, a_{l+1})$ for $l= i \log r +t-1$.
Assign to each node $b_i$ the (common) color of the edges $(b_i,b_j)$, $j>i$.
There are $r(n-1) +1$ distinct nodes $b_0, \ldots, b_{r(n-1)}$,
each assigned one of $r$ colors, therefore at least $n$
of them are assigned the same color.
These nodes induce a monochromatic clique of size $n$.
\qed

\medskip

Finally, consider how Ramsey relates to PWPP. It was shown previously in \citep{Komargodski2019} that there exists a randomized reduction from PWPP to {\sc Ramsey},  as well as a deterministic reductions from PWPP to the multi-color Ramsey problem. In the appendix, we use properties of metric spaces to provide an alternative deterministic reduction from PWPP to multi-color Ramsey.

\section{Sunflowers}
An important aspect in extremal combinatorics is the extremal \textit{bound}: how large a system has to be to guarantee that the desired combinatorial structure exists. Often times, the tightest bounds are unknown - improving these bounds is an important research tradition in Combinatorics \citep{Alweiss2020} \citep{ConlonFerber2020}. 

However, despite this uncertainty, we can still use weaker extremal bounds to define provably hard TFNP search problems, which we can then relate to the subclasses of TFNP. In this section, we focus on the Sunflower Lemma \citep{jukna}, but this paradigm for reasoning about extremal problems can be applied to many other problems as well.

We begin with some basic definitions.

\begin{defi}
	A {\em $k$-set system} is a collection of distinct sets in which every set contains exactly $k$ elements. A collection of distinct sets $S_1, \dots, S_n$ is a {\em sunflower} if for all $i$, $j$, $S_i \cap {S_j}$ is the same.
\end{defi}

The Sunflower lemma states that for all positive integers $k, s$, there is a number $f(k,s)$ such that every $k$-set system of size $f(k,s)$ contains a sunflower of size $s$. The lemma was formulated and proved by Erd\H{o}s and Rado \cite{ErdosRadoSunflower} for $f(k,s)=k!(s-1)^{k+1}$, using an inductive proof that applies an iterative pigeonhole argument.
The conjecture is that $f(k,s) \leq C^k$ for some constant C that depends only on $s$. Progress on improving the upper bound on $f(k,s)$ was made recently in \cite{Alweiss2020}.

A computational problem based on the Sunflower Lemma was formulated in \citep{Komargodski2019},
and shown to be hard on average assuming the existence of collision resistant hash function.
Here, we use an even weaker bound to define a problem which we call {\sc Naive Sunflower}, and relate it to the multi-color Ramsey problem.

\begin{defi}[{\sc Naive Sunflower}]
	We are given a poly($k$)-sized circuit 
	$$C: \{0,1\}^{k^3 \log k} \mapsto (\{0,1\}^{k^3 \log k})^k $$
	which is supposed to specify a family of $2^{k^3 \log k}$ distinct sets of size $k$ over a universe of $2^{k^3 \log k}$ elements
	(each element is represented by
	 a $k^3 \log k$-bit string).   The problem is to find either: 
	\begin{enumerate}
		\item (An error): An index $i$ such that the set (represented  by) $C(i)$ contains two identical elements, or find two distinct indices $i$, $j$ such that the sets $C(i), C(j)$ are equal,or
		\item A sunflower of size $k^2$.
	\end{enumerate}
\end{defi}

Following the exact same argument presented in \citep{Komargodski2019}, this problem is hard on average assuming that Collision Resistant Hash Function families exist.

We will reduce {\sc Naive Sunflower} to Multi-color Ramsey by using a characterization of large sunflowers as a pairwise equidistant collection of points in a metric space. Given a $k$-set system $F$, define the distance $d(A,B)$ between any two sets $A, B$ of $F$ as $d(A,B) = |A \Delta B|/2$, where $A \Delta B = \{A\setminus{B} \}\cup \{B\setminus{A}\}$ is their symmetric difference
(it has even size since $A, B$ have the same size).
The function $d$ is a valid metric.
Clearly, any sunflower in $F$ is a set of pairwise equidistant points in this metric.

Conversely, by a result of Deza \citep{DezaEquidistant}, any collection of at least $k^2-k+2$ pairwise equidistant $k$-sets must form a sunflower. Therefore, in this regime, we can reduce the problem of finding a sunflower of size $k^2$ to the problem of finding $k^2$ pairwise equidistant sets. In turn, we can reduce this problem to the multi-color Ramsey problem.

\begin{theorem}
    {\sc Naive Sunflower} reduces to {\sc $(\sqrt{n},n)$-Ramsey}.
\end{theorem}

\begin{proof}
    Consider an instance of {\sc Naive Sunflower}, given by a circuit $C$. As above, $C$ is supposed to define a set system consisting of $2^{k^3 \log k}$ sets, where each set contains $k$ elements. The core idea behind our reduction is simple: we construct a graph consisting of $2^{k^3 \log k}$ nodes, where the $i$-th node $u_i$ corresponds to the set $C(i)$. To color the edges of this graph, we use $k$ colors, given by the numbers $1, 2, \dots, k$. Every edge $(u_i,u_j)$ is colored as follows. If $C(i), C(j)$  are distinct $k$-sets, then assign color $d(C(i),C(j))$ to the edge $(u_i,u_j)$.
    If one of $C(i), C(j)$ is not a $k$-set, e.g. contains a duplicate element, or if the sets are equal, then assign color 1 to the edge $(u_i,u_j)$.
    
    Let $n = k^2$. Since the graph has 
    $2^{k^3 \log k} = k^{nk}$ nodes, it contains a monochromatic clique of size $n = k^2$
    by Ramsey's Theorem. Let $M$ be any such monochromatic clique of size $k^2$.
    If $M$ contains a node $u_i$ such that $C(i)$ is not a $k$-set, or if it contains two nodes $u_i, u_j$ such that $C(i), C(j)$ are equal sets, then we have a violation for the circuit $C$.
    Otherwise, the collection $\{ C(i) | u_i \in M \}$ has $k^2$ pairwise equidistant sets, and thus by Deza's theorem, they form a sunflower.
\end{proof}

We have shown earlier that {\sc ($r,n$)-Ramsey} is in PLC for any $r,n $; thus we can conclude: 

\begin{corollary}
    {\sc Naive Sunflower} is in PLC.
\end{corollary}

\section{Short Choice}

As we have seen above, PPP is contained in the class PLC. There is an intuitive reason for this: implicit in any PPP-complete problem is an iterated pigeonhole argument.

The class PEPP, introduced in \citep{Kleinberg2020TFNP}, embodies the dual of the class PPP - an anti-pigeonhole principle: if there are $2^n - 1$ pigeons and $2^n$ holes, then no matter how the pigeons are placed, there must be an empty hole. PEPP belongs to the class \textbf{TF$\Sigma_{2}$P}, which is believed to lie outside of the class $NP$.

While the existence proof for Pigeonhole Circuit has a majority argument, the existence proof for Empty has a corresponding ``minority" argument. Suppose we are given an instance of the PEPP-complete problem \textit{EMPTY}: we are given a poly(n)-sized circuit $C: [2^n - 1] \mapsto [2^n] $,
where the inputs and outputs are all represented concisely using exactly $n$ bits. The challenge here is to find an element $j$ in the range such that there is no $i$ with $C(i) = j$. There must exist a bit $c_1$ such that the minority of elements in the domain of $C$ map to a $n$-bit string whose first bit is $c_1$. This minority has size \textit{at most} $2^{n-1} - 1$. Among the elements in this minority, there must exist a bit $c_2$ such that the minority of these elements map to a $n$-bit string whose second element is $c_2$. This new minority has size \textit{at most} $2^{n-2} - 1$. If we continue this argument $n$ times, we find that there exists a bit string $c_1 \circ c_2 \circ \dots \circ c_n$ which is not in the image of $C$.

As discussed earlier, {\sc Long Choice} is a generalization of PPP that encapsulates the iterated majority arguments. We can define an analogous generalization of the class PEPP that encapsulates the iterated minority argument. For this, we introduce a problem called {\sc Short Choice}, a problem in \textbf{TF$\Sigma_{2}$P} which is the dual of {\sc Long Choice}. 

\begin{defi}[Subcertificate]
    Given a Long Choice problem instance defined by predicate functions $P_0, P_1, \dots, P_{n-2}$ we call a subcertificate a sequence of distinct elements $a_0, a_1, \dots, a_k$ ($k \leq n$) which satisfy the {\sc Long Choice} conditions imposed by the predicate functions $P_i$. That is, for each predicate function $i < k$, we require that $P_i(a_0, \dots a_{i}, a_j) $ 
    is the same for all $j > i$. 
\end{defi}

\begin{defi}[Problem: {\sc Short Choice}]
	The input is the same as in the {\sc Long Choice} problem, except that the universe $U$ has  now $2^n - 2$ objects. As in {\sc Long Choice}, we are given a sequence of $n-1$ poly($n$)-sized circuits, $P_0, \dots, P_{n-2}$, where each $P_i$ defines a predicate function $ P_i : U^{i+2} \mapsto \{0,1\}$ 
	
	The problem is to find  a sequence $a_0, a_1, \dots, a_{k}$ of at most $ n-1$ distinct objects in $U$ and a bit $c \in \{0, 1\}$ with the property that (1) the sequence $a_0, a_1, \dots, a_{k}$ is a subcertificate, and (2)  there does not exist any object $a_{k+1} \in U$ that both extends this subcertificate and has $P_{k}(a_0, a_1, \dots, a_{k}, a_{k+1}) = c$. 
\end{defi}

As in our discussion of {\sc Long Choice}, we can show that {\sc Short Choice} is both total and PEPP-hard. The proof of totality uses a repeated minority argument, in the same way that the proof for Long Choice used a repeated majority argument. And the proof of PEPP-hardness is along similar lines as the PPP-hardness proof for {\sc Long Choice}.

\begin{theorem}
    {\sc Short Choice} is a total problem.
\end{theorem}

\begin{proof}
    Let $U$ denote the set of $2^n - 2$ objects in the universe of this Short Choice problem instance. We will construct a certificate, given by a subcertificate $a_0, a_1, \dots, a_{k}$, $(k \leq{n-2})$ and a bit $c$. In order to construct this sequence, we will also define a sequence of nonempty sets $U_0 \supset U_1 \supset \dots \supset U_{k}$ with the following properties: 
    \begin{enumerate}
        \item $U_0 = U$
        \item $|U_k| \leq 2^{n - k} - 2$ for all $k$
        \item $a_k \in U_k$ for all $k$ \textit{and} whenever $k\geq{1}$, $a_0, \dots, a_{k-1} \not\in U_{k}$.
        \item For every $j \leq k$, an element $x$ extends the subcertificate $a_0, a_1, \dots, a_j$ if and only if $x\in U_j \setminus \{a_j \}$. 
        \item For all $x\in{U_{j+1}}$, $P_j(a_0, \dots, a_{j}, x)$ is the same.
        
    \end{enumerate}
    
    As a base case, we begin by defining $U_0 := U$, and we pick an arbitrary element $a_0\in{U_0}$. Note that the base case of our construction so far satisfies properties $1$ and $2$: $U_0 = U$, $|U_0| = |U| \leq 2^{n - 0} - 2$. Property $3$ holds because $a_0 \in U_0$ (the second condition in property $3$ is vacuously true here). Property $4$ holds trivially, since any element $x \neq a_0$ must belong to $U_0 \setminus a_0$. 
    
    We now define $U_1$ in such a way that our base case satisfies property $5$ as well. We partition the elements $x\in \{U_0\setminus{ \{a_0 \}  } \}$ in two groups based on the value $P_0(a_0, x)$. We know that $|U_0\setminus{ \{a_0 \}  }| = 2^n - 3 $, and since this set is being partitioned into two disjoint sets, the minority must have size at most $2^{n-1} - 2$. If the minority is nonempty, we define $U_1$ to be the minority (which guarantees condition $5$ holds). We also define $a_{1}$ as an arbitary element of $U_1$. Otherwise, if the minority is empty, this means that $P_0(a_0, x)$ takes on a constant value (call it $b$) for all $x\in U\setminus \{a_0 \}$. This in turn means that there is no value $x$ with $P_0(a_0, x) = \neg b$ (the opposite of bit $b$). We can thus return $a_0$ along with the bit $\neg b$ as a valid certificate to the Short Choice problem. This completes the base case. 
    
    Now, suppose for some $0 \leq j \leq {k}$ we have defined the nonempty sets $U_0 \supset U_1 \supset \dots \supset U_j$ and the subsequence $a_0, a_1, \dots, a_{j}$ in such a way that they satisfy the above properties. We can first conclude based on properties $3$ and $5$ that the sequence $a_0, \dots, a_j$ is a valid subcertificate. Furthermore, an element $x$ extends this subcertificate if and only if it belongs to $U_j\setminus \{ a_j \}$, according to property $4$.

    To continue our inductive construction, we partition the elements $x\in \{ U_j \setminus \{a_j \} \}$ based on the  value of $P_j(a_0, \dots, a_j, x)$. If one side of this partition is empty, that means that there exists a bit $c\in{ \{0, 1\} }$ such that there are no elements $x\in \{ U_j \setminus \{a_j \} \}$ where $a_0, a_1, \dots, x$ is a subcertificate and $P_j(a_0, \dots, a_j, x) = c$. In this case, we are done: the sequence $a_0, \dots, a_j$ along with the bit value $c$, serves as a certificate to the problem. 
    
    Otherwise, both sides of this partition are nonempty. In this case, we can continue the inductive construction: we pick the minority side of the partition of the elements $x\in \{ U_j \setminus \{a_j \} \}$. We then define $U_{j+1}$ to be all of the elements on the minority side of the partition, and pick an arbitrary element $a_{j+1}\in{U_{j+1}}$ to extend the  subcertificate. We know by the inductive hypothesis that $|U_j| \leq 2^{n -j} - 2$, and by the same Pigeonhole argument used in the base case, we know that $|U_{j+1}| \leq 2^{n -j} - 2$, as desired. Thus properties $1$ and $2$ hold. To see why property $3$ holds, note that $a_{j+1} \in U_{j+1}$. Furthermore, $a_0, \dots, a_{j-1} \not\in U_j \supset U_{j+1}$. Finally, $a_j \not\in U_j$, as stated above. To see why property $4$ must hold, consider any candidate element $a_{j+2}$ which may extend the Long Choice subcertificate. We know from the inductive hypothesis that $a_{j+2}$ must belong to $U_j \setminus a_j$. Additionally, we must also now have that

    $$P_j(a_0, \dots, a_j, a_{j+1}) = P_j(a_0, \dots, a_{j}, a_{j+2})$$
    
    By the definition of $P_j$, this means that $a_{j+1}$ and $a_{j+2}$ both belong to $U_{j+1}$, and since $a_{j+1}$ and $a_{j+2}$ must be distinct in order to be in the same subcertificate, we conclude $a_{j+2}\in U_{j+1}\setminus a_{j+1}$. This argument also shows that property $5$ continues to hold as well. 
    
    Note that our construction can only proceed until $j = n-2$. To see why, suppose that $j = n-2$. Then, by our inductive hypotheses, $|U_j| \leq 2^{n-(n-2)} - 2 = 2$. Furthermore, since $a_j \in U_j$ by our assumption, we know that $U_j \setminus \{ a_j \}$ has at most $1$ other element. Denote this element (if it even exists) by $a_{n-1}$, and consider the value $c_{bad} = P_{n-2}(a_0, \dots, a_{n-2}, a_{n-1})$. It is clear that if we consider $c$ to be the opposite bit of $c_{bad}$, there are no elements $x$ which extend this subcertificate with $P_{n-2}(a_0, \dots, a_{n-2}, x) = c$. Thus, we can return $a_0, \dots, a_{n-2}$ and $c$ and we are done. 
    
    Therefore, our construction is guaranteed to terminate with a subcertificate $a_0, \dots, a_j$ of the appropriate length, as well as a bit $c$, which provide us with a solution to the Short Choice problem.  
\end{proof}

\begin{theorem}
   {\sc  Short Choice} is PEPP-hard
\end{theorem}

\begin{proof}
    Consider an instance of the PEPP-complete problem, {\sc Empty}, which is given by a poly(n)-sized circuit $C$: 
    
    $$C: [2^n-2] \mapsto [2^n-1] $$
    
    where the challenge is to find an element, $e\in{[2^n-1]}$ such that for all $i\in{[2^n - 2]}$, $C(i) \neq {e}$.

    We now define a {\sc Short Choice} instance whose solution allows us to recover a solution to the {\sc Empty} problem.
    
    The universe, $U$, of this instance consists of the elements in the set $[2^n - 2]$ (the domain of $C$). It now suffices to define a sequence of predicate functions, $P_0, \dots, P_{n-2}$. 

    Each predicate function $P_i$ takes as input the distinct elements $a_0, a_1, \dots, a_i, x$ and follows a similar procedure to the proof of PPP-hardness of {\sc Long Choice},. First, it calculates a set $H_i$ of elements belonging to the range of $C$. This set $H_i$ has the form $[x_i, y_i]\setminus \{C(a_0), C(a_1), \dots, C(a_i) \}$. $P_i$ then calculates the \textbf{midpoint} of $H_i$, defined as the smallest value $k\in{H_i}$ such that half of the elements of $H_i$ are less than or equal to $k$. $P_i$ returns $1$ if $C(x)$ is less than or equal to the midpoint of $H_i$ and $0$ otherwise.
    
    It now remains to define the sets $H_i$. The sets $H_i$ are defined inductively with the following two key properties:
    
    \begin{enumerate}
        \item $|H_i| \geq 2^{n-i} - 2$
        \item If $a_0, a_1, \dots, a_i$ is a Long Choice subcertificate, then for every $k \in \{0, 1, \dots, i-1 \}$, $C(a_j)\in H_{k} \bigcup \{C(a_0), C(a_1), C(a_2), \dots, C(a_k) \}$ whenever $j > k$. 
    \end{enumerate}
    
    As a base case, given a first input of $a_0$, $P_0$ defines $H_0 := [2^n - 1] \setminus \{ C(a_0) \}$. This satisfies the first inductive property: $[2^n - 1] \setminus \{ C(a_0) \}$ has $2^{n-0} - 2$ elements. 
    
    To see why $H_0$ satisfies the second inductive property, note that $H_0 \bigcup C(a_0)$ is actually the set $[2^n - 1]$ - the range of $C$! Therefore, for any  subcertificate, $a_0, a_1, \dots, a_i$, we will have, for all $k \in \{1, \dots, i\}$, that $C(a_k) \in H_0 \bigcup C(a_0)$, as desired.

    Now, suppose that the $H_0, \dots, H_i$ have been defined in such a way that they satisfy the two inductive properties. Suppose $P_i$ takes as input a valid  subcertificate $a_0, \dots, a_i$, along with a final element $a_{i+1}$. Then, $P_i$ first uses the elements $a_0, \dots, a_i$ to define $H_i$. Next, it calculates the midpoint of $H_i$. Then, $P_i(a_0, a_1, \dots, a_i, a_{i+1}$ returns $1$ if $C(x)$ is less than or equal to the midpoint and $0$ otherwise. Accordingly, we define $H_{i+1}$ as follows:
    \begin{enumerate}
        \item If $P_i(a_0, \dots, a_i, a_{i+1})$ is $1$, then we define $H_{i+1}$ to be the elements of $H_i\setminus{C(a_{i+1})}$ which are less than or equal to the midpoint of $H_i$ 
        \item Otherwise, we define $H_{i+1}$ to be the elements of $H_i\setminus{C(a_{i+1})}$ which are greater than the midpoint of $H_i$
    \end{enumerate}
    
    To prove that the first inductive property holds for $H_{i+1}$, recall that by our inductive assumption, $H_{i}$ contains at least $2^{n-i} - 2$ elements. By definition, the midpoint will split this set into two sets, $A$ and $B$, each of size at least $2^{n-i-1} - 1$. $H_{i+1}$ is defined as either $A \setminus C(a_{i+1})$ or $B \setminus C(a_{i+1})$. Thus, in either case, $H_i$ has at least $2^{n-i-i} - 2$ elements, as desired. 
    
    To prove that the second property holds, assume without loss of generality that $C(a_i)$ is greater than the midpoint of $H_{i-1}$. Then, as described above, we define $H_i$ to be the subset of $H_{i-1}\setminus C(a_i)$ containing elements greater than the midpoint of $H_{i-1}$. 
    
    Now, in order for $P_{i-1}(a_0, a_1, \dots, a_{i-1}, a_j)$ to be constant for all $j \geq i$, we must have that $C(a_j)$ is also greater than the midpoint of $H_{i-1}$. However, by the second inductive assumption, we know that $C(a_j)$ must belong to $H_{i-1} \bigcup \{C(a_0), \dots, C(a_{i-1}) \}$. If we apply these two facts together, we conclude that $C(a_{j})$ must belong to the set $H_{i} \bigcup \{a_0, \dots, a_{i-1}, a_i \}$. Letting $j = i+1$ in our case proves that the second inductive property holds.

    Finally, consider any certificate to this instance. It consists of a  subcertificate $a_0, \dots, a_j$ ($j \leq n-2$) and a bit, $c$. As we know, this bit $c$ has the following property: there is no object $a_{j+1}$ which both extends the certificate and also has $P_{j}(a_0, \dots, a_{j}, a_{j+1}) = c$. 
    
    Consider the set $H_{j}$; it has size at least $2^{n - (j-2)} - 2 \geq 2^{n - (n-2)} - 2 = 2$. Thus, the midpoint of $H_{j}$ splits $H_{j}$ into two nonempty subsets of size at least $1$. Let $A$ denote the subset of $H_{j}$ containing elements less than or equal to its midpoint, and let $B$ denote the subset of $H_{j}$ containing elements greater than its midpoint. Suppose that $c = 0$. Then we can conclude that there are no elements $C(a_{j+1}) \in B$. If there were, we could pick such an element to extend the Long Choice sequence. Thus any element of $B$ would be a solution to our problem. Similarly, if $c = 1$; then we can similarly conclude that there are no elements $C(a_{j+1}) \in A$; thus any element of $A$ would be a solution to our problem. Finally, note that we can easily identify which elements belong to the set $A$ and $B$; as mentioned earlier, they take the form $[x_i, y_i]\setminus \{C(a_0), C(a_1), \dots, C(a_i) \}$. This completes the proof. 
\end{proof}

We can define the class PSC (Polynomial Short Choice) to be the class whose complete problem is {\sc Short Choice}.

\section{K\"onig and Erd\H{o}s-Ko-Rado}

In this section we introduce and characterize computational problems associated with two classical theorems in combinatorics.

\subsection{K\"onig}

K\"onig's lemma states that in every infinite connected graph with finite degree there is an infinite simple path starting at every node.
The lemma is often stated and used for trees: Every infinite (rooted) tree with finite branching has an infinite path (starting at the root).
The finite version of the lemma is that every large enough connected graph (or tree) with bounded degree contains a long path. For example, every rooted binary tree with $2^n$ nodes contains a path of length $n$.
The graph version follows easily from the tree version: Given a connected graph, take a spanning tree of the graph, for example a breadth-first-tree from an arbitrary node.

The standard proof of K\"onig's tree lemma (both the infinitary as well as the finitary version) is by the same type of repeated majority argument as in the proof of totality for Long Choice. 
Starting from the root of the tree, proceed to the
child whose subtree contains the largest number of nodes, and repeat the process from there.
If the tree is infinite, one of the children must have an infinite subtree (since the degree is finite), thus this process will generate an infinite path. Similarly, in the finite case, every iteration reduces the number of nodes at most by a factor of $d$ (the degree), so the process generates a path of logarithmic length.

Given a (succinctly represented) exponentially large connected graph or tree with bounded degree, e.g. a binary tree, how hard is it to find a long simple path? We formulate this problem below for binary trees, represented through the parent information.  

\begin{defi}[Problem: \sc{K\"onig}]
	We are given a poly(n)-sized circuit, $P:\mathbb \{0,1\}^{n}\to\mathbb \{0,1 \}^{n} \Cross \{0,1\} $.
	This circuit is supposed to define a rooted binary tree on $2^n$ nodes, where each node is encoded by a $n$-bit string. It does so by defining a parent relation: for a node $u$, $P(u)$ is an ordered pair $(v, b)$ where $v$ is the (binary encoding of) the parent of $u$, and $b$ is a bit which indicates whether $u$ is the left or right child of $v$. 
	If $P(u) = u$, that means that $u$ does not have a parent (it is a root node). In the K\"onig problem, we are given $P$ and a root node, $r$, and are asked to find either a violation ($P$ does not specify a binary tree rooted at $r$) or find a path of length $n$. Specifically, return one of the following certificates: 
	
	\begin{enumerate}
		\item \textbf{Identical children: } Return $2$ distinct nodes $a$ and $b$ with the property that $P(a) = P(b)$. (That is, they are both left children or right children of the same node). 
		\item \textbf{Invalid Root: } Return $r$ if $P(r)\neq{r}$.
		\item \textbf{Non-Unique Root: } Return a node $s\neq{r}$ with the property that $P(s) = s$.
		\item \textbf{Far Away Node: } Return a node $s$ with the following property: if we apply the parent operator $P$ to $s$ a total of $n$ times, we do not reach the root $r$. 
		\item \textbf{Long Path: } A sequence of $n+1$ nodes $a_0, a_1, \dots, a_{n}$ with $a_0 = r$, and the property that $P(a_i) = a_{i-1}$ for all $i\geq{1}$. Note that it suffices to provide $a_{n-1}$ as a valid certificate here; the rest of the path can be recovered by applying the parent operator $P$. 
	\end{enumerate}
\end{defi}	
We refer to the circuit $P$ as the {\em parent operator}. 
In cases 1, 2, 3, $P$  does not induce a binary tree. The same is true in case 4, if applying $P$ to $s$ $n$ times produces a repeated node; otherwise, we get a simple path of length $n$. Case 5 yields a path of length $n$ from the root $r$.

The proof of K\"onig's lemma suggests that the problem should be in PLC. It turns out that {\sc K\"onig}
is in fact in PPP, and furthermore it is complete.

\begin{theorem}
    {\sc K\"onig} is PPP-complete.
\end{theorem}

We show first the hardness:

\begin{lemma}
	{\sc K\"onig}  is PPP-hard
\end{lemma}

\begin{figure}[t!]
	\centering
	\includegraphics[width=0.5\textwidth]{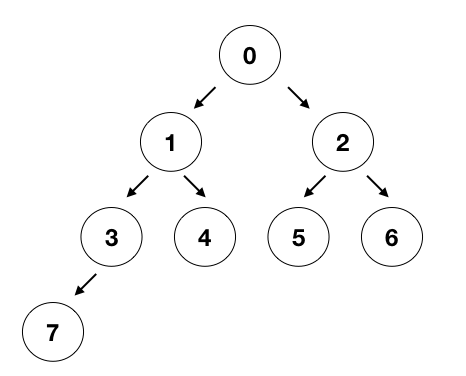}
	\caption{\textbf{Indexing of a Binary Tree}. Here we provide an illustration of how we index a binary tree in our proof that Konig is PPP-hard. An arrow from node ``a" to node ``b" indicates that ``a" is the parent of ``b". At the top of the image, the root has encoding $0$. Its left child is $1$ and its right child is $2$. Given a node, $x$, the parent function $P$ calculates $P(x) = \lfloor \frac{x-1}{2} \rfloor$.}
	\label{fig:IndexingBinaryTree}
\end{figure}

\begin{proof}
    We  provide a reduction from {\sc Collision}. Suppose we are given a circuit $C$ that defines a {\sc Collision} instance: 
    
    $$C: \{0,1\}^n \mapsto \{0,1\}^n $$
    
    We define a {\sc K\"onig} problem instance on $2^{n+1}$ nodes, by a parent mapping : 
    
    $$P: \{0, 1\}^{n+1} \mapsto \{0,1\}^{n+1} \Cross \{0,1\}$$
    
    Note that $P$ implicitly defines a graph on $2^{n+1}$ nodes. The nodes are supposed to be arranged in a binary tree with root $0$.
    We start by defining the positions of nodes whose binary encodings are in the range $[0, 2^n - 1]$, and we do so recursively. First, $0$ is the root. Next, consider any node $s$ in this range. If the binary encoding of $s$ is odd, we define $P(s) = (\lfloor\frac{s-1}{2} \rfloor, 0)$; that is, $s$ is the left child of a node with binary encoding $\lfloor{\frac{s-1}{2}}\rfloor$. Similarly, if the binary encoding of $s$ is even, we define $P(s) = (\lfloor{\frac{s-1}{2}}\rfloor, 1)$. See Fig.~\ref{fig:IndexingBinaryTree} for an illustration of this structure in the case that $n = 3$. (This arrangement is similar to the indexing of a heap.) Note that so far, our definition of $P$ has absolutely no dependence on the circuit $C$. 
    
    To finish the definition of $P$, it remains to consider the nodes whose binary encodings are in the range $A = [2^n, 2^{n+1} - 1]$. Consider a given node $s$ whose encoding lies in the range given by $A$. Then $s - 2^n$ lies in the domain of $C$. Furthermore, for each $s\in{A}$, $s-2^n$ corresponds to a unique element in the domain of $C$. For each $s\in{A}$, we let 
    $r = C(s - 2^n) + 2^{n} - 1$ and we define 
    $$P(s) = (\lfloor \frac{r-1}{2} \rfloor, \text{parity}(C(s-2^n)))$$
    
    where the parity function returns $0$ if the input is even and $1$ if it is odd.
    
    Because the circuit $C$ has range $[0, 2^n-1]$, $r$ must lie in the interval $[2^{n}-1, 2^{n+1}-1]$. Then, based on the definition of $P(s)$, the parent of each $s\in{A}$ must lie in the interval $B = [2^{n-1}-1, 2^{n}-1]$. There are $2^{n-1}$ nodes in the interval $B$, and each of these nodes can have at most $2$ children. Thus, a total of $2^n$ nodes can have parent nodes in the interval $B$.
    
    On the other hand, every node of $A$ must have a parent in the interval $B$, and we know there are $2^n$ nodes in the interval $A$. In addition to the nodes of $A$, we know that the node $2^{n}-1$, which does not belong to $A$, also has a parent in $B$: $P(2^{n}-1) = 2^{n-1}-1$. This implies that a total of $2^n+1$ elements must have a parent belonging to the interval $B$. 
    
    Since $B$ can only have $2^n$ children, by the pigeonhole principle, there must exist a node in $B$ with two left children or two right children. By the definition of the {\sc K\"onig} problem, these two left or right children form a valid certificate to the {\sc K\"onig} problem. It remains to show that these two nodes also provide a certificate for the {\sc Collision} problem. 
    
    There are two cases to consider. In the first case, suppose that parent node $2^{n-1}-1$ has two left children. We know that one of these children is the node $2^{n}-1$, and the other child must be a node $s\in{A}$. In this case, if $P(s) = (\lfloor \frac{r-1}{2} \rfloor = 2^{n-1}-1, 0)$, we can conclude that $r = 2^{n}-1$. However, this would in turn imply that $C(s - 2^n) = 0$, which means $s - 2^n$ is a zero element to our original {\sc Collision} problem! 
    
    In the second case, suppose that a parent node $i > 2^{n-1}-1$ has two left (or right) children. This would imply that there are two distinct nodes $s_1, s_2\in{A}$ with $P(s_1) = P(s_2)$. By the definition of $P$, this in turn implies that $C(s_1 - 2^n) = C(s_2 - 2^n)$. Since $s_1 \neq{s_2}$, we can conclude that $s_1 - 2^n$ and $s_2 - 2^n$ provide us with a collision in our {\sc Collision} certificate. 
\end{proof}

\begin{lemma}
	{\sc K\"onig} is in PPP
\end{lemma}

\begin{proof}
    We are given an instance of the {\sc K\"onig} problem, which is defined by a root node, $r$, and a circuit $P$:
    
    $$P: \{0, 1\}^{n} \mapsto \{0,1\}^{n} \Cross \{0,1\}$$
    
    As described above, $P$ implicitly describes a graph on $2^{n}$ nodes. We wish to reduce it to an instance of Pigeonhole Circuit in polynomial time. To do this, we define a circuit, $C$:
    
    $$C: \{0, 1\}^{n} \mapsto \{0,1\}^{n} $$
    
    where the domain of $C$ will be the nodes of the graph defined by $P$.

    In any binary tree, we can assign every node a unique ``index" based on its position relative to the root. This indexing scheme is defined in an inductive fashion. First, the root is given index $0$. Next, suppose a given node $a_i$ has index $i$. Then, we say that its left child has index $2*i + 1$ and its right child has index $2*i + 2$; see Fig.~\ref{fig:IndexingBinaryTree} for an illustration of this indexing scheme. 
    
    Consider any node of the graph, $g$. Suppose that $g$ is neither a Far away node nor a Long path certificate nor a Non-unique root. (Note that all of these conditions are easy to verify). Under these assumptions, we show how to efficiently find the index of $g$. To do this, we repeatedly apply the parent operator, $P$, to the node $g$ until we reach the root node. Based on the assumptions we have made, this is always possible. Furthermore, it will require at most $n-1$ applications of the parent operator. Every time we apply the parent operator to a node, we receive two pieces of information: the node's parent and whether the node is a left or right child of its parent. Thus, once we reach the root node in this process, we have a sequence of nodes $a_0, a_1, a_2, \dots, a_k$ where $a_k = r$ and $a_0 = g$. Furthermore, for each $a_i$ and $a_{i+1}$, we know whether $a_i$ is the left or right child of $a_{i+1}$. We therefore have a path from $r$ to $g$. Using this path, we can use the indexing scheme mentioned above to find the indices of the nodes $r = a_k, a_{k-1}, \dots, a_0 = g$ in that order to finally arrive at the index of node $g$. 
    
    We are now in a position to fully define the circuit $C$. On input $i$, $C(i)$ outputs: 
    
    \begin{enumerate}
        \item $0$ if $r$ is an invalid root
        \item $0$ if $i\neq{r}$ and $i$ is a root (in this case, there is a non-unique root).
        \item $0$ if $i$ is a Far Away Node
        \item $0$ if $i$ provides a certificate for a Long Path.
        \item If \textbf{none} of the above conditions are met, then let $x$ be the index of node $i$ in the binary tree. $C(i)$ returns $x+1$.
    \end{enumerate}
    
    It now remains to show that a certificate to the {\sc K\"onig} problem can be recovered from a certificate to the {\sc Collision} problem we have just defined. There are two cases to consider. 
    
    In the first case, suppose our certificate to the {\sc Collision} problem is a zero element. In this case, there are four possibilities. Either $r$ is an invalid root, or $i\neq{r}$ is a root, or $i$ is a Far Away Node or $i$ is a certificate for a Long Path. We can polynomially verify all of these conditions, and they all represent valid certificates to the {\sc K\"onig} problem. 
    
    In the second case, suppose our certificate to the {\sc Collision} problem is a collision of two elements in which neither element is a zero element: $i\neq{j}$ with $C(i) = C(j)\neq{0}$. Then, we know that $i$ and $j$ are not Long Path certificates and are not Far Away nodes. Furthermore, we know that their indices must be equal to each other, which in turn implies that they are left (or right) children of the same parent node. Thus, $i$ and $j$ are identical children and provide us with a valid certificate.
\end{proof}

Theorem 8 follows from the above two lemmas.

\subsection{Erd\H{o}s-Ko-Rado}

The Erd\H{o}s-Ko-Rado Lemma is one of the foundational results in extremal set theory.

\begin{theorem}[Erd\H{o}s-Ko-Rado Lemma]
If $F$ is any $k$-set system over a universe $X$ of size $n > 2k$, and every pair of sets in $F$ has non-empty intersection, then
	$$|F| \leq \binom{n-1}{k-1}$$
\end{theorem}

Thus, if $|F| > \binom{n-1}{k-1} $ then $F$ must contain two disjoint sets. How hard is it to \textit{find} these two disjoint sets, if $F$ is a succinctly given exponentially large system? 

Take the case in which $k = 2$ and the size of the universe is $2^n$. In this case, the Erdos-Ko-Rado lemma tell us that the largest possible intersecting set system has size $\binom{2^n - 1}{1} = 2^n-1$. Therefore, given a $2$-set system $F$ of size $2^n > 2^n - 1$, then $F$ must contain two disjoint sets. 

\begin{defi}[Problem Statement: {\sc Erd\H{o}s-Ko-Rado}]
   We are given a poly($n$)-sized circuit
   $$F: \{0,1\}^n \mapsto (\{0,1\}^n)^2 $$
   which is supposed to represent a 2-set system of size $2^n$ over the universe $X=\{0,1\}^n$.
   The problem is to find either a violation ($F$ is not a valid encoding) or two disjoint sets of the set system. Specifically, return one of the following certificates: 
	\begin{enumerate}
    \item (Error:)	An index $i$ such that $F(i)=(a,a)$ for some $a \in \{0,1\}^n$ (i.e., the set $F(i)$ has two identical elements), or two distinct indices $i ,j$ such that the sets (represented by) $F(i), F(j)$ are equal, or
    \item (Disjoint sets:) Two indices $i, j$ such that the sets (represented by) $F(i), F(j)$ are disjoint.
	\end{enumerate} 
\end{defi}

By the Erd\H{o}s-Ko-Rado lemma, one of these conditions must occur, placing the problem in TFNP. We now show that this problem is equivalent to PPP. 

\begin{theorem}
	{\sc Erd\H{o}s-Ko-Rado} is PPP-complete. 
\end{theorem}

\begin{proof}
	We first prove that {\sc Erd\H{o}s-Ko-Rado} is PPP-hard. Accordingly, suppose we are given some instance of {\sc Collision}, defined by a circuit 
	$C: \{0,1\}^n \mapsto \{0,1\}^n $,
	and the challenge is to find an input that maps to $0$ or find two distinct inputs which map to the same value. 
	
	We construct a circuit, $F$, which implicitly defines a $2$-set system of size $2^n$. Our circuit $F$ takes, as input, an $n$-bit string, and outputs a $(\{0,1\}^n)^2$-bit string: $F(x) = (0^n, C(x))$.

	It remains to show that any solution to this instance of the {\sc Erd\H{o}s-Ko-Rado} problem allows us to recover a solution to the original {\sc Collision} problem. There are $3$ possible types of certificates that we can find in the {\sc Erd\H{o}s-Ko-Rado} problem instance. First, consider a certificate providing two disjoint sets. This is impossible from the above definition of $F$, since every set contains the element $0^n$. Second, consider a certificate $x$ providing an invalid set $F(x)$. This means that $C(x)=0^n$, i.e. we find a zero element in the {\sc Collision}  problem. Third, consider a certificate which defines a repeat set. That is, we have two indices $a$ and $b$ with $F(a)$ and $F(b)$ defining the same set. As we see in the above definition of $F$, this implies that $C(a) = C(b)$, which provides us with a collision for our original {\sc Collision} problem. Thus, any certificate for the {\sc Erd\H{o}s-Ko-Rado} problem provides us with a certificate for the {\sc Collision} problem above. 
	
	We now reduce {\sc Erd\H{o}s-Ko-Rado}  to {\sc Collision}. Consider an {\sc Erd\H{o}s-Ko-Rado} problem instance, defined by a circuit $F$ which implicitly provides us with a $2$-set system of size $2^n$. We construct a circuit $C: \{0,1\}^n \mapsto \{0,1\}^n$ as follows. 
	We first examine two arbitrary sets from the set system, $F$. Without loss of generality, we examine $F(0)$ and $F(1)$. If these two sets are disjoint, identical, or invalid, we are done. Otherwise, they have an intersection of exactly $1$ element. Let $F(0) = (a,b)$ and $F(1) = (b,c)$, such that their intersection is $\{b\}$. For any given index $i$, define $C(i)$, as follows: 
	
	\begin{enumerate}
		\item If $F(i)$ outputs an invalid set, set $C(i) = 0$
		\item If $F(i)$ contains both $b$ and $d$ (where $d\neq{b}$), then we let $C(i) = d$.
		\item If $F(i)$ does not contain $b$, then set $C(i) = 0$. 
	\end{enumerate}
	
	We now demonstrate that a solution to the {\sc Collision} problem for circuit $C$ allows us to recover a certificate to the original {\sc Erd\H{o}s-Ko-Rado} problem. 
	
	There are two cases to consider here. We start by considering the case in which we recover a zero element. That is, we find an element $x$ with $C(x) = 0$. This can only happen in $2$ cases, based on the above definition of $C$.  First, $F(x)$ might generate an invalid set. In this case, $x$ is a valid certificate. Second, $F(x)$ might not contain $b$. In this case, suppose $F(x) = (p,q)$. We note that $x$ is distinct from $1$ and $0$, because $0$ and $1$ are not zero elements, based on the above definition. If $F(x)$ is disjoint from either $F(1)$ or $F(0)$, then we are done. If, instead, $F(x)$ does not contain $b$ and has nonempty intersection with $F(1)$ and $F(0)$, then $F(x)$ must contain $a$ and $c$. Then, consider \textit{any} index $y$ that is distinct from $0,1,x$. If $F(y)$ is invalid, we are done. If $F(y)$ is identical to one of $F(0), F(1), F(x)$, we have found a repeated set certificate, and we are done. Otherwise, $F(y)$ must be fully disjoint from at least one of $F(0), F(1), F(x)$, and we can recover two disjoint sets. That is, the only way a set $F(y)$ of size 2 intersects each of $\{a,b\}, \{a,c\}, \{b,c\}$ is if $F(y)$ is identical to one of $\{a,b\}, \{a,c\}, \{b,c\}$. 
	
	Now, suppose that we recover a collision from $C$. That is, we find two distinct $n$-bit strings $x$ and $y$ with $C(x) = C(y)$. First, suppose $C(x) = C(y) \neq{0}$. Based on the above definition of $C$, this can only happen if $F(x)$ and $F(y)$ both contain the element $b$ as well as the element $C(x)$. In this case, $F(x)$ and $F(y)$ both describe the set $\{b, C(x)\}$, which means that we have recovered a "repeated set" certificate. In the second case, suppose $C(x) = C(y) = 0$. That is, suppose we have found two zero certificates. In this case, the argument from above tells us that we can recover the desired certificate using any one of these zero certificates.  In this case, $F(x)$ and $F(y)$ are each either invalid sets or they do not contain $b$. In the case that they are invalid sets, we are done. If these indices both do not contain $b$, then $F(x)$ and $F(y)$ must each contain both $a$ and $c$ in order to have nonempty intersection with $F(1)$ and $F(2)$.

\end{proof}

\section{Mantel, Tur\'an, and Bad Colorings}


In this section, we introduce a new flavor of problems from extremal combinatorics which generalize PPP. This class of problems is related to Mantel's and Tur\'an's theorem and graph colorings.
Mantel's theorem states that a triangle-free graph with $N$ nodes has at maximum $\lfloor N^2/4 \rfloor$ edges. The maximum is achieved by a complete bipartite graph with equal or almost equal parts (depending on whether $N$ is even or odd).
Tur\'an's theorem answers the generalized question of what is the maximum number of edges in a graph with $N$ nodes that does not contain a $(k+1)$-clique: the maximum is achieved by a complete $k$-partite graph that has equal or almost equal parts  
(see \citep{jukna} for a detailed exposition). 

The same quantities answer the easier question of, what is the maximum number of edges of a $k$-colorable graph on $N$ nodes.
A $k$-colorable graph whose color classes have sizes $x_1, \ldots, x_k$ can have at most $\sum_{i \neq j} x_i x_j$ edges. Since the $x_i$'s are integers that sum to $N$, it can be shown that the maximum is achieved when they are all equal or almost equal.

These theorems induce corresponding total computational problems: Given (succinctly) an exponential graph with more edges than the above bounds of Mantel or Tur\'an, find a triangle or a $(k+1)$-clique respectively.
If we are given in addition a $k$-coloring of the nodes, find an illegally colored edge.
We call these problems respectively {\sc Mantel},
{\sc $k$-Tur\'an} and {\sc Bad $k$-Coloring} 
({\sc Mantel} is just {\sc $2$-Tur\'an}).
We define below formally the problems as TFNP problems.

\begin{defi}[{\sc $k$-Tur\'an}]
   We are given a poly($n$)-sized circuit $E: [ \binom{k}{2} (2^n)^2+1] \mapsto [(k 2^n)^2]$, which is supposed to represent a graph with $k 2^n$ nodes
   and $\binom{k}{2} (2^n)^2+1$ edges (nodes and indices of edges are encoded by bit-strings of appropriate length as usual); $E$ maps the index of an edge to the two nodes of the edge. The problem is to find either a violation ($E$ is not a valid encoding of the edges of a graph) or a $(k+1)$-clique. Specifically, return one of the following certificates:
   	\begin{enumerate}
 		\item (Error): An index $i$ such that $E(i)$ consists of two identical nodes, or two distinct indices $i,j$ such that $E(i), E(j)$ contain the same two nodes
 		(not necessarily in the same order), or
 		\item ($(k+1)$-clique): $\binom{k+1}{2}$ indices which are mapped by $E$ to the edges of a clique on $k+1$ nodes.
	\end{enumerate}
\end{defi}

\begin{defi}[{\sc Bad $k$-Coloring}]
    We are given a poly($n$)-sized circuit $E: [ \binom{k}{2} (2^n)^2+1] \mapsto [(k 2^n)^2]$, which is supposed to represent a graph with $k 2^n$ nodes
   and $\binom{k}{2} (2^n)^2+1$ edges and a poly($n$)-sized circuit $C: [k 2^n] \mapsto [k] $ which colors the nodes with $k$ colors. The problem is to find either a violation ($E$ is not a valid encoding of the edges of a graph) or an edge whose nodes have the same color. Specifically, return one of the following certificates:
   	\begin{enumerate}
 		\item (Error): An index $i$ such that $E(i)$ consists of two identical nodes, or two distinct indices $i,j$ such that $E(i), E(j)$ contain the same two nodes
 		(not necessarily in the same order), or
 		\item (Bad edge): An index $i$ such that $E(i) = (a,b)$ and $C(a)=C(b)$.
	\end{enumerate}
\end{defi}

We show the following relations between PPP and these problems:

\begin{theorem}
1. PPP reduces to {\sc Bad $2$-Coloring}.\\
2. For all $k \geq 2$, {\sc Bad $k$-Coloring} reduces to {\sc $k$-Tur\'an}. In particular, {\sc Bad $2$-Coloring} reduces to {\sc Mantel}.\\
3. For all $k \geq 2$, {\sc Bad $k$-Coloring} reduces to {\sc Bad $(k+1)$-Coloring}.\\
4. For all $k \geq 2$, {\sc $k$ Tur\'an} reduces to {\sc $(k+1)$-Tur\'an}.
\end{theorem}

\begin{proof}
1. We reduce from the {\sc Collision} problem. Given a circuit $D: [0,2^n-1] \mapsto [0,2^n-1]$  for the  {\sc Collision} problem, we construct an instance $(E,C)$ of the {\sc Bad $2$-Coloring} problem
on $2^{n+1}$ nodes $V=[0,2^{n+1}-1]$. The coloring $C$ maps nodes $[0,2^n -1]$ to color 0 and nodes $[2^n, 2^{n+1}-1]$ to color 1. The edge function $E: [0,2^{2n}] \mapsto V^2$ is defined as follows. For an index $i < 2^{2n}$, let $a_i = \lfloor \frac{i}{2^n} \rfloor $ and $b_i = i \mod 2^n$; we set $E(i) = (D(a_i), D(b_i) + 2^n)$.  For index $i=2^{2n}$ we set $E(2^{2n})= (0, 2^n)$.

We claim that a certificate for {\sc Bad $2$-Coloring} readily yields a certificate for the {\sc Collision} instance. Note first that every edge $E(i)$ consists of a node in $[0,2^n -1]$ and a node in $[2^n,2^{n+1}-1]$, and these nodes have different colors. Therefore, the only possible certificate for the {\sc Bad $2$-Coloring} instance is two distinct indices $i, j$ such that $E(i)=E(j)$.
If both $i,j < 2^{2n}$, then $(D(a_i), D(b_i) + 2^n)= (D(a_j), D(b_j) + 2^n)$, hence $D(a_i)=D(a_j)$ and $D(b_i)=D(b_j)$. Since $i \neq j$, either $a_i \neq a_j$ or $b_i \neq b_j$ (or both), thus we get a solution to our original {\sc Collision} instance. If one of $i,j$ is $2^{2n}$, say $i=2^{2n}$ and $j< 2^{2n}$, then $D(a_j)=0$ and $D(b_j)=0$, thus we get an element that is mapped to 0 by $D$.

\medskip
\noindent
2. Given an instance $(E,C)$ of {\sc Bad $k$-Coloring} which defines a graph $G$ and a $k$-coloring $C$ of its nodes, consider the instance of {\sc $k$-Tur\'an} specified by the same edge circuit $E$.
A certificate for this {\sc $k$-Tur\'an} instance either gives an error in the function $E$, which is also a certificate  for the {\sc Bad $k$-Coloring} instance, or consists of the indices of the edges of a clique on $k+1$ nodes in $G$. At least two of these $k+1$ nodes are given the same color by $C$, thus the edge connecting them is a certificate for the {\sc Bad $k$-Coloring} instance.

\medskip
\noindent
3. Let $(E,C)$ be an instance of the {\sc Bad $k$-Coloring} problem specifying a graph $G$ on $k 2^n$ nodes $V=[0, k 2^n -1]$. The function $E: [0, \binom{k}{2} (2^n)^2] \mapsto V^2$ specifies the edges of the graph and the function $C: V \mapsto [k]$ specifies a coloring of the nodes with $k$ colors.
We construct an instance $(E',C')$ of {\sc Bad $(k+1)$-Coloring} that specifies a graph $G'$ on $(k+1) 2^n$ nodes $V' = V \cup W$ where $W=[k 2^n, (k+1)2^n -1]$. The coloring function $C'$ maps every node $u \in V$ to its original color $C(u) \in [k]$ and maps every node $u \in W$ to color $k+1$.
The edge set of $G'$ consists of all the edges of $G$ and all possible edges between $V$ and $W$.
Note that $\binom{k}{2} (2^n)^2 +1 + k2^n \cdot 2^n = \binom{k+1}{2} (2^n)^2 +1$.
We define the function $E'$ so that it maps the first $\binom{k}{2} (2^n)^2 +1$ indices to the edges of $G$ (i.e. set $E'(i) = E(i)$ for all $i \in [0, \binom{k}{2} (2^n)^2]$), and maps the remaining indices to distinct pairs $(v,w), v \in V, w \in W$.

Consider a certificate for the {\sc Bad $(k+1)$-Coloring} instance $(E',C')$. All new edges in $V \times W$ are distinct valid edges that are legally colored with different colors. Therefore, the certificate must consist of one or two  original edges of the given graph $G$, and thus it is alao a certificate for the given {\sc Bad $k$-Coloring} instance.

\medskip
\noindent
4. The reduction is the same as in part 3. A $(k+2)$-clique in $G'$ is either entirely contained in $G$ or it consists of a node of $W$ and a $(k+1)$-clique in $G$.
\end{proof}

Thus, we have a hierarchy of problems on top of PPP.  The Bad Coloring problems can be viewed as instances of the pigeonhole problem 
(there is no iteration here), but the mapping is given indirectly and cannot be easily constructed:
We can view the indices of the edges as the pigeons and the potential legal edges, i.e. all the pairs of differently colored nodes, as the holes. There are more pigeons that holes, so either two pigeons are mapped to the same hole
($E(i), E(j)$ are the same edge for some pair of indices $i,j$)), or some pigeon is not mapped to a hole (for some $i$, $E(i)=(a,a)$ or $E(i)=(a,b)$ with $C(a)=C(b)$; this corresponds to the special 0 value in the PPP problem). The difference with PPP, is that the set of holes (the range of the mapping) is not given a priori explicitly 
as a set of bit-strings (or integers) as in PPP,
but rather it is implied indirectly by the coloring $C$.
As a consequence, even for $k=2$, we cannot compute easily in polynomial time for example the number $h$ of available holes ($h$ is the product of the sizes of the two color classes), and we cannot compute efficiently an index function mapping each legal pair of nodes (pair $(a,b)$ with $C(a) \neq C(b)$) to an index in $[h]$. 
In the Tur\'an problems, there is in addition the complication of optimizing over all partitions (colorings) and of seeking a clique rather than a single edge.

Another example along the same lines is the following  {\sc bad $k$-set coloring} problem; Given (poly($n$)-size circuits specifying) a $k$-coloring $C$ of a set $V$ of $k 2^n$ nodes and a family $F$ of $2^{kn}+1$ $k$-sets over $V$, find a $k$-set in $F$ that is not panchromatic, i.e. two of its elements have the same color, or find two equal sets in $F$. The case $k$=1 is equivalent to the {\sc Collision} problem (which defines PPP). For higher values of $k$, the problems form a hierarchy, where again existence of a certificate is guaranteed by (1) the answer to an optimization problem (what is the maximum number of panchromatic $k$-sets over all $k$-colorings), and (2) the pigeonhole principle, where however the mapping is not given explicitly, but is defined indirectly in an inefficient manner. See the appendix for detailed proofs of these properties of {\sc bad $k$-set coloring}.

\section{Discussion and Future Work}

The generalizations of PPP which we explore in this paper seem to give rise to a remarkably rich set of tantalizing open questions.  Some examples:

\begin{enumerate}
  \item Prove black box separations between the problems and classes studied in this paper. For example, prove a separation between PPP and PLC; between PEPP and PSC; between PPP and the hierarchy  of Tur\'an and Bad coloring problems.
    \item What other natural problems belong to PLC or are PLC-complete? One problem in TFNP that has long evaded classification is {\sc Bertrand-Chebyshev}: given a number $n$, find a prime number between $n$ and $2n$. 
    \item What is the complexity of finding monochromatic cliques in smaller graphs, whose existence is guaranteed by a century of fascinating improvements of Ramsey's theorem?
    \item What natural problems belong to the class PSC (but not to PEPP)? 
    \item How does PLC relate to problems in cryptography, and specifically lattices? 
    \item The Bad Coloring hierarchy suggests a novel source of computational hardness: inefficient encoding of objects. What other interesting natural problems share this type of hardness? 
    \item More generally, what other problems from extremal combinatorics give rise to search problems in TFNP, and how are these problems classified in TFNP subclasses? 
\end{enumerate}

\newpage

\bibliography{references} 

\begin{thebibliography}{15}
\providecommand{\natexlab}[1]{#1}
\providecommand{\url}[1]{\texttt{#1}}
\expandafter\ifx\csname urlstyle\endcsname\relax
  \providecommand{\doi}[1]{doi: #1}\else
  \providecommand{\doi}{doi: \begingroup \urlstyle{rm}\Url}\fi

\bibitem[Alweiss et~al.(2020)Alweiss, Lovett, Wu, and Zhang]{Alweiss2020}
Ryan Alweiss, Shachar Lovett, Kewen Wu, and Jiapeng Zhang.
\newblock Improved bounds for the sunflower lemma.
\newblock In \emph{Proceedings of the 52nd Annual ACM SIGACT Symposium on
  Theory of Computing}, STOC 2020, page 624–630, 2020.

\bibitem[Bollob\`as(2013)]{bollobas}
B\`ela Bollob\`as.
\newblock \emph{Extremal Graph Theory}.
\newblock Dover, 2013.

\bibitem[Conlon and Ferber(2020)]{ConlonFerber2020}
David Conlon and Asaf Ferber.
\newblock Lower bounds for multicolor ramsey numbers, 2020.
\newblock URL \url{https://arxiv.org/abs/2009.10458}.

\bibitem[Deza and Frankl(1981)]{DezaEquidistant}
Michel Deza and Peter Frankl.
\newblock Every large set of equidistant (0, +1, −1)-vectors forms a
  sunflower.
\newblock \emph{Combinatorica}, 1:\penalty0 225--231, 09 1981.

\bibitem[Erdös and Rado(1960)]{ErdosRadoSunflower}
P.~Erdös and R.~Rado.
\newblock Intersection theorems for systems of sets.
\newblock \emph{Journal of the London Mathematical Society}, s1-35\penalty0
  (1):\penalty0 85--90, 1960.

\bibitem[Goldberg and Papadimitriou(2018)]{GOLDBERG2018}
Paul~W. Goldberg and Christos~H. Papadimitriou.
\newblock Towards a unified complexity theory of total functions.
\newblock \emph{Journal of Computer and System Sciences}, 94:\penalty0 167 --
  192, 2018.

\bibitem[Graham et~al.(1990)Graham, Rothschild, and Spencer]{graham1990ramsey}
Ronald~L Graham, Bruce~L Rothschild, and Joel~H Spencer.
\newblock \emph{Ramsey theory}, volume~20.
\newblock 'John Wiley \& Sons', 1990.

\bibitem[Jeřábek(2016)]{JERABEK2016380}
Emil Jeřábek.
\newblock Integer factoring and modular square roots.
\newblock \emph{Journal of Computer and System Sciences}, 82\penalty0
  (2):\penalty0 380 -- 394, 2016.

\bibitem[Johnson et~al.(1988)Johnson, Papadimitriou, and
  Yannakakis]{JOHNSON1988}
David~S. Johnson, Christos~H. Papadimitriou, and Mihalis Yannakakis.
\newblock How easy is local search?
\newblock \emph{Journal of Computer and System Sciences}, 37\penalty0
  (1):\penalty0 79 -- 100, 1988.

\bibitem[Jukna(2013)]{jukna}
Stasys Jukna.
\newblock \emph{Extremal Combinatorics With Applications in Computer Science}.
\newblock Springer Berlin, 2013.

\bibitem[Kleinberg et~al.(2021)Kleinberg, Korten, Mitropolsky, and
  Papadimitriou]{Kleinberg2020TFNP}
Robert Kleinberg, Oliver Korten, Daniel Mitropolsky, and Christos~H.
  Papadimitriou.
\newblock Total functions in the polynomial hierarchy.
\newblock In \emph{12th Innovations in Theoretical Computer Science Conference,
  {ITCS}}, volume 185 of \emph{LIPIcs}, pages 44:1--44:18, 2021.

\bibitem[Komargodski et~al.(2019)Komargodski, Naor, and Yogev]{Komargodski2019}
Ilan Komargodski, Moni Naor, and Eylon Yogev.
\newblock White-box vs. black-box complexity of search problems: Ramsey and
  graph property testing.
\newblock \emph{J. ACM}, 66\penalty0 (5), July 2019.

\bibitem[Megiddo and Papadimitriou(1991)]{MEGIDDO_TFNP_1991}
Nimrod Megiddo and Christos~H. Papadimitriou.
\newblock On total functions, existence theorems and computational complexity.
\newblock \emph{Theoretical Computer Science}, 81\penalty0 (2):\penalty0
  317--324, 1991.

\bibitem[Papadimitriou(1994)]{PAPADIMITRIOU1994}
Christos~H. Papadimitriou.
\newblock On the complexity of the parity argument and other inefficient proofs
  of existence.
\newblock \emph{Journal of Computer and System Sciences}, 48\penalty0
  (3):\penalty0 498 -- 532, 1994.

\bibitem[Sotiraki et~al.(2018)Sotiraki, Zampetakis, and Zirdelis]{Sotiraki2018}
Katerina Sotiraki, Manolis Zampetakis, and Giorgos Zirdelis.
\newblock Ppp-completeness with connections to cryptography.
\newblock In \emph{59th {IEEE} Annual Symposium on Foundations of Computer
  Science, {FOCS}}, pages 148--158. {IEEE} Computer Society, 2018.

\end{thebibliography}

\newpage

\appendix

\section{Missing material from Section 2 (Long Choice)}

We prove that in the definition of Long Choice, it is unimportant what the initial element $a_0$ actually is, and whether it is specified or not. Consider a variant of the problem where a specific initial element is required. 

\begin{defi}[Problem: {\sc Constrained Long Choice}]
	Consider a set $U$ of $2^n$ objects, each represented by a unique binary $n$-bit string. We are given a sequence of $n-1$ predicate functions, $P_0, \dots, P_{n-2}$ represented by poly($n$)-size circuits. Predicate function $P_i$ has arity $i + 2$: 
	$$ P_i : U^{i+2} \mapsto \{0,1\}$$
	We are also given an initial element, $a_0$. 
	The problem is to find a sequence of $n+1$ distinct objects $a_0, \dots, a_{n}$ in $U$, with the following property: for all $i$ in $[0, \dots, n-2]$, for all $j > i$, $P_i(a_0, \dots, a_{i},a_{j})$ is the same.

\end{defi}

By the proof of Theorem 1, {\sc Constrained Long Choice} is also a total search problem.

\begin{proposition}
	{\sc Long Choice} with no initial element is equivalent to {\sc Constrained Long Choice}.
\end{proposition}

\begin{proof}
	It is clear that {\sc Long Choice} with no initial element reduces to {\sc Constrained Long Choice}. For any instance of the former problem, we can arbitrarily specify an initial element, turning the problem into an instance of the latter problem. 
	
	In the other direction, suppose we are given a {\sc Constrained Long Choice} instance which specifies an initial element, $a_0$ and a sequence of predicate functions $P_0, P_1, \dots, P_{n-1}$. 
	
	Accordingly, we define a {\sc Long Choice} instance with no initial element. To do so, we define a new sequence of predicate functions $T_0, T_1, \dots, T_{n-2}$. $T_k$ has the same arity as $P_k$, and it is defined in terms of $P_k$. For distinct $b_0, b_1, \dots, b_{k+1}$, $T_k(b_0, b_1, \dots, b_{k+1})$ is defined using the following sequence of operations. 
	\begin{enumerate}
		\item If $b_0 = a_0$, then $T_k(b_0, b_1, \dots, b_{k+1}) = P_k(b_0, b_1, \dots, b_{k+1})$. In this case, $T_k$ and $P_k$ are identical. 
		\item Otherwise, $b_0 \neq{a_0}$. Consider the input to $T_k$, a sequence $b_0, b_1, \dots, b_{k+1}$. We perform the following operation: we first replace $b_0$ with $a_0$. Then, we replace any instances of $a_0$ among $b_1, \dots, b_{k+1}$ with $b_0$. In some sense, we have 'swapped' $a_0$ and $b_0$. This yields a modified sequence $a_0, c_1, \dots, c_{k+1}$. We define $T_k(b_0, b_1, \dots, b_{k+1}) = P_k(a_0, c_1, \dots, c_{k+1})$. 
	\end{enumerate}
	
	Now, consider any certificate for this problem, given by a sequence of distinct elements $d_0, d_1, \dots, d_{n}$. We perform the same 'swap' operation from step $2$. That is, we replace $d_0$ with $a_0$ and we replace any instance of $a_0$ among $d_1, \dots, d_n$ with $d_0$. This yields a modified sequence $a_0, e_1, \dots, e_n$, which serves as a valid certificate to the original Constrained Long Choice problem. 
	
	To see why this is true, we can consider two cases (as above). First, if $d_0 = a_0$, then, we are done, since $T_k$ and $P_k$ are identical in this case, so our certificate automatically serves as a certificate for the original problem. Otherwise, if $d_0\neq{a_0}$, all functions $T_k$ consider the modified sequence $a_0, e_1, \dots, e_{k+1}$. In this scenario, for any $k\in{\{0,1, \dots, n-1\}}$, $T_k(d_0, d_1, \dots, d_{k+1}) = P_k(a_0, e_1, \dots, e_{k+1})$. If $d_0, d_1, \dots, d_{n}$ is indeed a valid certificate, we know that for any $k\in{\{0, 1, \dots, n-2\}}$, $T_k(d_0, d_1, \dots, d_{k}, d_{k+1}) = T_k(d_0, d_1, \dots, d_k, d_{j})$ for all $j > k$. This, in turn, implies the same fact for $P_k$. Namely, it implies that for any $k\in{\{0, 1, \dots, n-2\}}$, $P_k(a_0, e_1, \dots, e_{k}, e_{k+1}) = T_k(a_0, e_1, \dots, e_k, e_{j})$ for all $j > k$. Thus, $a_0, e_1, \dots, e_n$ is a valid certificate to the original constrained Long Choice problem, and we are done. 
\end{proof}

Finally, we define Long Choice with Order formally:

\begin{defi}[Long Choice with Order]
	Consider a set $U$ of $2^n$ objects, each represented by a unique $n$-bit string. We are given a sequence of $n-1$ predicate functions, $P_0, \dots, P_{n-2}$ represented by poly($n$)-size circuits. Predicate function $P_i$ has arity $i + 2$: 
	$$ P_i : U^{i+2} \mapsto \{0,1\}$$
We are also given a function $F$, which purportedly defines a strict total order over the above set $U$. The function $F$ is also  represented by a poly($n$)-size circuit: 
	
	$$F: U^2 \mapsto \{0,1\} $$
	Given two distinct inputs $a_x, a_y$, $F(a_x, a_y) = 0$ indicates that $a_x < a_y$ in this total ordering. $F(a_x, a_y) = 1$ indicates that $a_x > a_y$. The problem is to find any of the following:
	\begin{enumerate}
		\item \textbf{Monotone certificate: }A monotone increasing sequence of $n+1$ distinct objects $a_0, \dots, a_{n}$ in $U$,  with the following property: for each $i$ in $[0, \dots, n-2]$, $P_i(a_0, \dots, a_{i},a_{j})$ is the same for all $j > i$.
		\item \textbf{Order Violation: } A set of $3$ distinct objects $a_x, a_y, a_z$ which violate the transitivity property of total orders.
	\end{enumerate}
\end{defi}

\section{Missing material from Section 3 (Ramsey)}

As mentioned earlier,  it was shown previously in \citep{Komargodski2019} that there exists a randomized reduction from PWPP to {\sc Ramsey}  as well as a deterministic reduction from PWPP to multi-color Ramsey. We provide here an alternative deterministic reduction from PWPP
to a multi-color Ramsey problem using properties of metric spaces. 

First, recall the well-known fact that the maximum number of pairwise equidistant points in $n$-dimensional Euclidean space $\mathbb{R}^n$ is $n+1$.
Consider the set $U = \{0,1\}^m$. Note that the Hamming and Euclidean distance metrics are in direct correspondence over $U$. In particular, the Hamming distance between two elements $i, j\in {U}$ is the \textit{square} of the Euclidean distance between those two elements. 
Therefore, there cannot be more than $m+1$ pairwise equidistant points in $U = \{0,1\}^m$, where distance is measured using the Hamming distance metric.
We can use this fact to reduce PWPP to multi-color Ramsey.

We will use below the shorthand term {\sc $k$-color standard Ramsey} to refer to the
$(k, \frac{n}{k \log k})$-color Ramsey problem. (We omit floor and ceiling functions to make things more readable.). Note that this problem is defined on a complete graph of size $2^n$; this is why we call it a standard problem. The edges are colored with $k$ colors and we seek a monochromatic clique of size $\frac{n}{k \log k}$.

\begin{theorem}
	{\sc $(n^{\delta})$-color standard Ramsey} is PWPP-hard for all $\delta < 1/2$
\end{theorem}

\begin{proof}
	Fix any positive $\delta < 1/2$. Recall that PWPP($n, n-1$) reduces to PWPP($n, n^{\delta}$).
	Consider an instance of the latter problem, given by a circuit, $C$: 
	
	$$C: \{0,1\}^{n} \mapsto \{0,1\}^{n^{\delta}} $$.
	
	We reduce this problem to a $n^{\delta}$-color standard Ramsey problem, which is defined by a circuit: 
	
	$$C_r: (\{0, 1\}^n)^2 \mapsto [n^{\delta}]$$
	
	$C_r$ is defined as follows: For distinct inputs, $i, j$, the circuit $C_r$: 
	\begin{enumerate}
		\item Returns the hamming distance between $C(i)$ and $C(j)$ if $C(i) \neq{C_j}$
		\item Returns $\lceil n^{\delta} / 2 \rceil$ if $C(i) = C(j)$. 
	\end{enumerate}

	As mentioned above, this problem is total and we are guaranteed to find a clique of size $\frac{n}{\delta n^{\delta \log n  }} = \frac{n^{1-\delta}}{\delta  \log n} \geq{\frac{n^{1 - \delta }}{\delta  n^{\epsilon}}} > 2 * n^{\delta}$, where $0<\epsilon<1/2 - \delta$ and all inequalities hold asymptotically. Thus, the guaranteed clique size is greater than $n^{\delta} + 1$. 
	
	We know that there are at most $n^{\delta} + 1$ distinct points in $\{0,1\}^{n^{\delta}}$ which are pairwise equidistant (under Hamming distance). 
	Consider a desired certificate (a clique) consisting of vertices $a_1, a_2, \dots, a_i$. What is the color of the edges of this clique? We claim that it must be the color $\lceil n^{\delta} / 2 \rceil$. To see why, suppose that it is any other color. Then, consider the $n^{\delta}$-bit strings $C(a_1), C(a_2), \dots, C(a_i)$. By the definition of $C_r$, this collection of values must be pairwise equidistant, which is impossible since $i > n^{\delta} + 1$.
	
	Thus, $C_r(a_x, a_y) = \lceil n^{\delta} / 2 \rceil$ for all $a_x, a_y$. Furthermore, there can be at most $n^{\delta} + 1$ distinct elements in the set $C(a_1), C(a_2), C(a_3), \dots, C(a_i)$. Thus, at least two elements of our clique collide under $C$, and we are done. 
\end{proof}

\section{Missing material from Section 7 (Mantel, Tur\'an, Bad Coloring) }


The {\sc Bad $k$-set coloring} problems constitute a similar hierarchy with PPP at is base.
We define first formally the problems.

\begin{defi}[{\sc Bad $k$-set Coloring}]
    We have a set $V=[0,k 2^n -1]$ of $k 2^n$ objects (encoded by bitstrings with $n+ \lceil \log k \rceil$ bits) and we are given a poly($n$)-sized circuit $C: V \mapsto [k]$ defining a $k$-coloring of the objects, and a poly($n$)-size circuit $F: [0,2^{kn}] \mapsto V^k$, which is supposed to represent a $k$-set system with $2^{kn} +1$ sets. The problem is to find one of the following certificates:
   	\begin{enumerate}
 		\item (Repeated set): Two distinct indices $i,j$ such that $F(i), F(j)$ contain the same elements
 		(not necessarily in the same order), or
 		\item (Bad set): An index $i$ such that $F(i)$ contains two elements $a ,b$ with  $C(a)=C(b)$; the two elements $a, b$ could be identical.
	\end{enumerate}
\end{defi}

\begin{theorem}
1.  PPP is equivalent to  {\sc Bad $1$-set Coloring}.\\
2. For all $k \geq 1$, {\sc Bad $k$-set Coloring} reduces to {\sc Bad $(k+1)$-set Coloring}
\end{theorem}
\begin{proof}
1.  This follows immediately from the definitions. Note that for $k=1$, all objects have the same color. Also an $1-$set is just a singleton, so it is not bad. Thus, {\sc Bad $1$-set Coloring} is simply the problem of finding a collision for a given mapping $F: [0,2^{n}] \mapsto [0,2^n-1]$, which is a PPP-complete problem.

\medskip
\noindent
2. Given an instance $(F, C)$ of {\sc Bad $k$-set Coloring} on a set $V=[0,k 2^n -1]$ of $k 2^n$ objects, we construct an instance $(F',C')$ of {\sc Bad $(k+1)$-set Coloring} on a set $V'=[0,(k+1) 2^n -1]$ of $(k+1) 2^n$ objects. For the coloring function $C'$, we let $C'(i)=C(i) \in [k]$ for $i \in V$, and $C'(i)=k+1$ for
$i \in W =[k 2^n, (k+1) 2^n -1]$. The function $F'$ that defines the $(k+1)$-set system combines every $k$-set in $F$ with every (new) object in $W$, except for the last $k$-set that is only combined with only one new object. Note that $2^{(k+1)n} +1 = 2^{kn} \cdot 2^n +1$.
Formally, for each index $i \in [0,2^{(k+1)n}]$, let $a_i = \lfloor \frac{i}{2^n} \rfloor$, and $b_i = i \mod 2^n$;
set $F'(i) = ( F(a_i), b_i + k2^{n})$. Note that $F'(2^{(k+1)n}) = ( F(2^{kn}), k 2^n)$ (this is the only set that combines $F(2^{kn})$ with an element of $W$).

Consider a certificate for the instance $(F',C')$ of {\sc Bad $(k+1)$-set Coloring}. If it is a bad set, i.e. an index $i$ such that $F'(i)$ contains two objects with the same color, then $F(a_i)$ must contain two elements with the same color, since all objects in $W$ have color $k+1$; thus, $a_i$ is a certificate for the original 
{\sc Bad $k$-set Coloring} instance. 

If the certificate for the instance $(F',C')$ of {\sc Bad $(k+1)$-set Coloring} is a repeated set, i.e. two distinct indices $i, j$ such $F'(i)$ and $F'(j)$ represent the same set,
then we must have that $F(a_i), F(a_j)$ represent the same subset of $V$, and $b_i=b_j$.
Since $i \neq j$ and $b_i = b_j$, we have $a_i \neq a_j$, and the indices $a_i, a_j$ are a certificate for
the original  {\sc Bad $k$-set Coloring} instance. 
\end{proof}

\end{document}